\documentclass[a4paper,12pt]{article} 
\usepackage[onehalfspacing]{setspace}
\usepackage{color}
\usepackage{enumitem}
\usepackage{amsmath,amssymb,amsthm}
\usepackage{subcaption}
\usepackage{graphicx}
\usepackage{algorithm}
\usepackage{algpseudocode}
\usepackage{color}
\usepackage{anysize}
\usepackage{multirow}
\usepackage{authblk}

\newtheorem{theorem}{Theorem}[section]
\newtheorem{lemma}[theorem]{Lemma}

\newcommand{\floor}[1]{\left\lfloor #1\right\rfloor}
\newcommand{\ceil}[1]{\left\lceil #1\right\rceil}

\title{Total Domination in Unit Disk Graphs}

\author{Sangram K. Jena \thanks{sangram@iitg.ac.in}}
\author{Gautam K. Das \thanks{gkd@iitg.ac.in}\thanks{corresponding author}}
\affil{Department of Mathematics\\Indian Institute of Technology Guwahati }

\begin{document}

\maketitle

\begin{abstract}
Let $G=(V,E)$ be an undirected graph. We call $D_t \subseteq V$ as a total dominating set (TDS) of $G$ if each vertex $v \in V$ has a dominator in $D$ other than itself. Here we consider the TDS problem in unit disk graphs, where the objective is to find a minimum cardinality total dominating set for an input graph.  We prove that the TDS  problem is NP-hard in unit disk  graphs. Next, we propose an 8-factor approximation algorithm for the problem. The running time of the proposed approximation algorithm is $O(n \log k)$, where $n$ is the number of vertices of the input graph and $k$ is output size. We also show that TDS problem admits a PTAS in unit disk graphs.

{\bf keywords:}
Total dominating set, approximation algorithm, PTAS, unit disk graph
\end{abstract}

\section{Introduction}
Let us consider a simple undirected graph $G=(V, E)$. The \emph{open neighbourhood} (\emph{resp. closed neighbourhood}) of a vertex $v \in V$ is the set $N_G(v) = \{u\in V: (u,v)\in E\}$  (resp. $N_G[v] = N_G(v) \cup \{v\}$).  A \emph{dominating set} (DS) of $G$ is a subset $D \subseteq V$ such that for each vertex $v\in V$, $D\cap N_G[v]\geq 1$. A \emph{total dominating set} (TDS) is a subset $D_t \subseteq V$ of $G$ such that for each vertex $v \in V$, $D_t \cap N_G(v) \geq 1$. Therefore,  a vertex $v \in D$ (dominating set) dominates all its neighbours and itself whereas a vertex $v \in D_t$ (total dominating set) dominates all its neighbours other than itself. The objective of TDS (resp. DS) problem is to find a minimum size subset $D_t \subseteq V$ (resp. $D \subseteq V$) such that $D_t$ (resp. $D$) dominates all the vertices in $V$.

The intersection graph of equal-radii disks in the plane is called a \emph{unit disk graph} (UDG). Given a set $S = \{d_1, d_2,\ldots, d_n\}$ of equal-radii $n$ circular disks in the plane, each with diameter 1, the corresponding UDG $G = (V, E)$ is defined as follows: each vertex $v_i \in V$ corresponds to the disk $d_i \in S$, and $(v_i, v_j) \in E$ if and only if the Euclidean distance between center of the disks $d_i$ and $d_j$ is less than or equal to 1.

\subsection{Related Work}
In 1980, Cockayne et al. \cite{cockayne1980} introduced the total domination problem and proved that for any connected graph $G$ of 
$n (\geq 3)$ vertices the cardinality of minimum total dominating set, denoted by  $\lambda_t$, is less than or equal to $\frac{2}{3}n$ i.e., $\lambda_t \leq \frac{2}{3}n$. 
Brigham et al.  \cite{brigham2000connected} proved that the total domination number is exactly $\frac{2}{3}n$ for the connected graph $G$ of order $n (\geq 3)$, where $G$ is either $C_3$ (cycle graph of 3 vertices), $C_6$ or 2-corona of some connected graph. 
Later, Sun \cite{sun1995upper} improved the bound to $\lambda_t \leq \floor{\frac{4}{7}(n+1)}$, for connected graphs having order $n$ with minimum degree at least 2. 
Chv{\'a}tal and McDiarmid \cite{chvatal1992small} and Tuza \cite{tuza1990covering} independently proved a theorem concerning transversals in hypergraphs, which gives a bound on total domination number. The bound is $\lambda_t\leq \frac{n}{2}$ for the graphs with order $n$ and minimum degree at least 3. 
For the graphs with minimum degree at least 4, Thomass{\'e} and Yeo \cite{thomasse2007total} proposed a result for hypergraphs, which bounds the total domination number by $\lambda_t\leq \frac{3}{7}n$. 
In \cite{delavicna2007some}, DeLaVi{\c{n}}a et al. proved that the total domination number of any connected graph is equal to the total domination number of a spanning tree of the same graph. 
Another interesting aspect of trees with respect to total domination is that it is possible to characterize some vertices that are in every total dominating set or not in any total dominating set \cite{cockayne2003vertices}. Furthermore, Haynes and Henning established three equivalent conditions for a tree to have a unique minimum total dominating set \cite{haynes2002trees}. 
Chellali and Haynes \cite{chellali2006note} proved $\lambda_t\geq \frac{n+2-\ell}{2}$ for a nontrivial tree of $n$ vertices with $\ell$ leaves. 
Dorfling et al. \cite{dorfling2006domination} bound the total domination number of planar graphs having different diameter and radius. 
Pfaff et al. \cite{pfaff1983np} showed that computing $\lambda_t$ for general graphs is NP-complete. 
In the same paper, they also showed that calculating $\lambda_t$ for bipartite graphs remains NP-complete. 
However, a linear time algorithm exists for computing $\lambda_t$ in tree graph \cite{laskar1984algorithmic}. 
The total domination number in case of star graphs, complete graphs, binary star graphs and complete bipartite graphs is 2 \cite{amos2012total}. 
In the same article, they have observed that for cycles and paths, total domination number can be calculated in polynomial time. 
They have also established a set of relations between  (i) $\lambda_t$ and the maximum degree, and (ii) $\lambda_t$ and the cut vertices of the graph. 
See \cite{amos2012total}, \cite{haynes1998fundamentals}, \cite{haynes1998domination}, \cite{henning2009survey},  \cite{henning2013total} for detailed survey on the TDS problem.

\subsection{Our Contribution} \label{contribution}
In this paper, we consider the total dominating set problem in unit disk graphs. In Section \ref{hardness}, we show that the decision version of the TDS problem is NP-complete in unit disk graphs. We propose an almost linear time 8-factor approximation algorithm in Section \ref{sec:apx}. We also show that the problem admits a PTAS in  Section  \ref{sec:ptas}. Finally, we conclude the paper in Section \ref{conclusion}.

\section{NP-Completeness} \label{hardness}
In this section, we show that the TDS problem in UDGs is NP-complete. The {\it vertex cover} (VC) problem in planar graph with maximum degree 3 is known to be NP-hard \cite{garey}. To prove NP-hardness result of the TDS problem in UDGs, we use polynomial time reduction from vertex cover problem in planar graph to it. Now, we define decision version of the TDS problem in UDGs and vertex cover problem in planar graphs as follows:

\begin{description}
 \item [The TDS problem in UDGs] (\textsc{Tds-Udg})
 \item [Instance:] A unit disk graph $G$ and an integer $k(>0)$.
 \item [Question:] Does $G$ has a TDS of size at most $k$?
\end{description}

\begin{description}
 \item [The VC problem in planar graphs] (\textsc{Vc-Pla})
 \item [Instance:] A planar graph $G$ with maximum degree 3 and an integer $k(>0)$.
 \item [Question:] Does $G$ has a VC of size at most $k$?
\end{description}
 
\begin{lemma}[\cite{mishra2020DAM}] \label{cor:embedding}
Let $G=(V,E)$ be a planar graph with maximum degree 3. The graph $G$ can be embedded in linear time on a planar grid of size $4 \times 4$ using $O(|V|^2)$ area such that the coordinate of each vertex $v \in V$ is $(4i,4j)$ for some integers $i, j$ and each edge $e \in E$ is a finite sequence of consecutive line segments of length 4 units along the grid lines. 
\end{lemma}

\begin{figure}[!ht]
  \centering
  \begin{minipage}{.45\textwidth}
  \hspace{-2.9cm}
  \centering
  \includegraphics[scale=0.73]{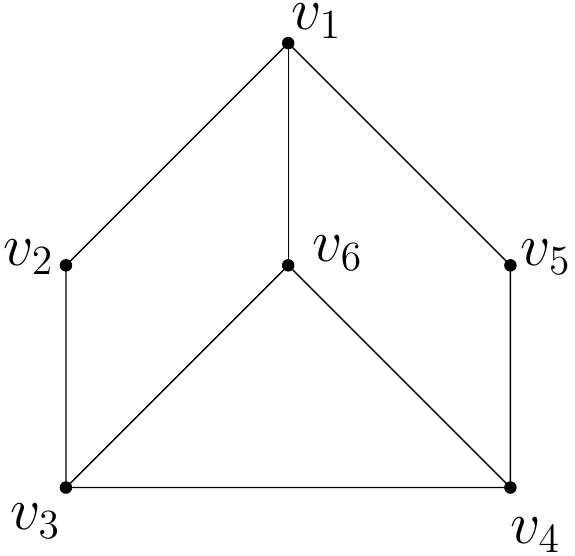} \\\hspace{-2.9cm} {(a)}
  \end{minipage}
  \begin{minipage}{.50\textwidth}
  \hspace{-1.9cm}
  \centering
  \includegraphics[scale=0.76]{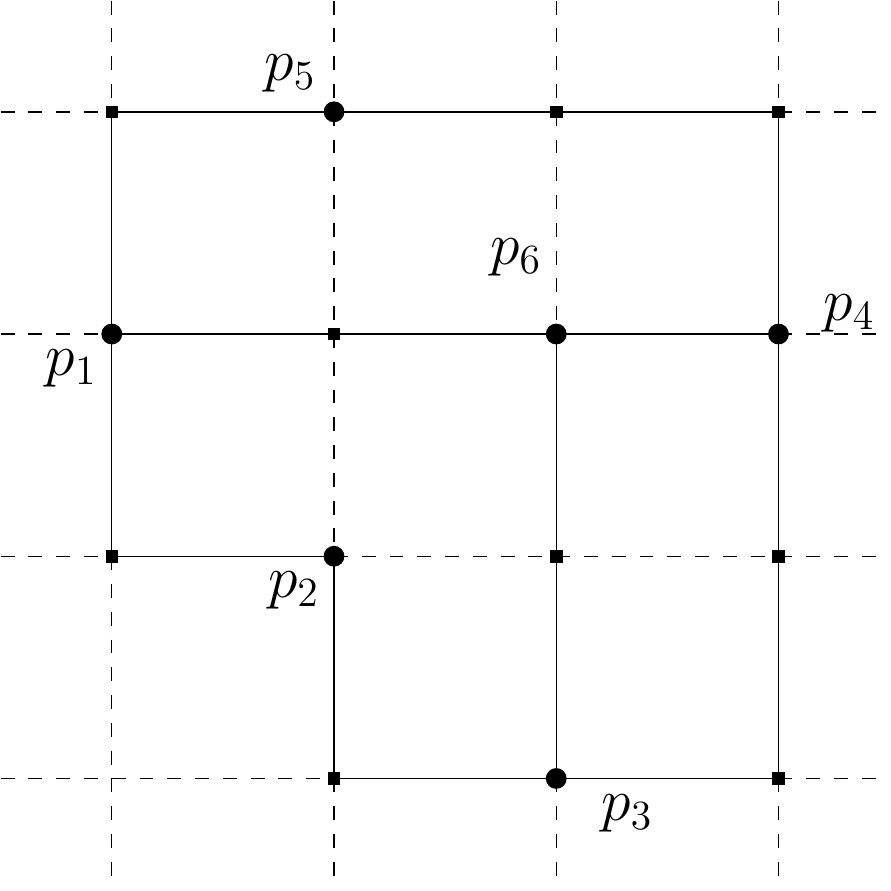}\\\hspace{-1.7cm}
  {(b)}
  \end{minipage}
  \begin{minipage}{.9\textwidth}
  \centering
  \includegraphics[scale=0.9]{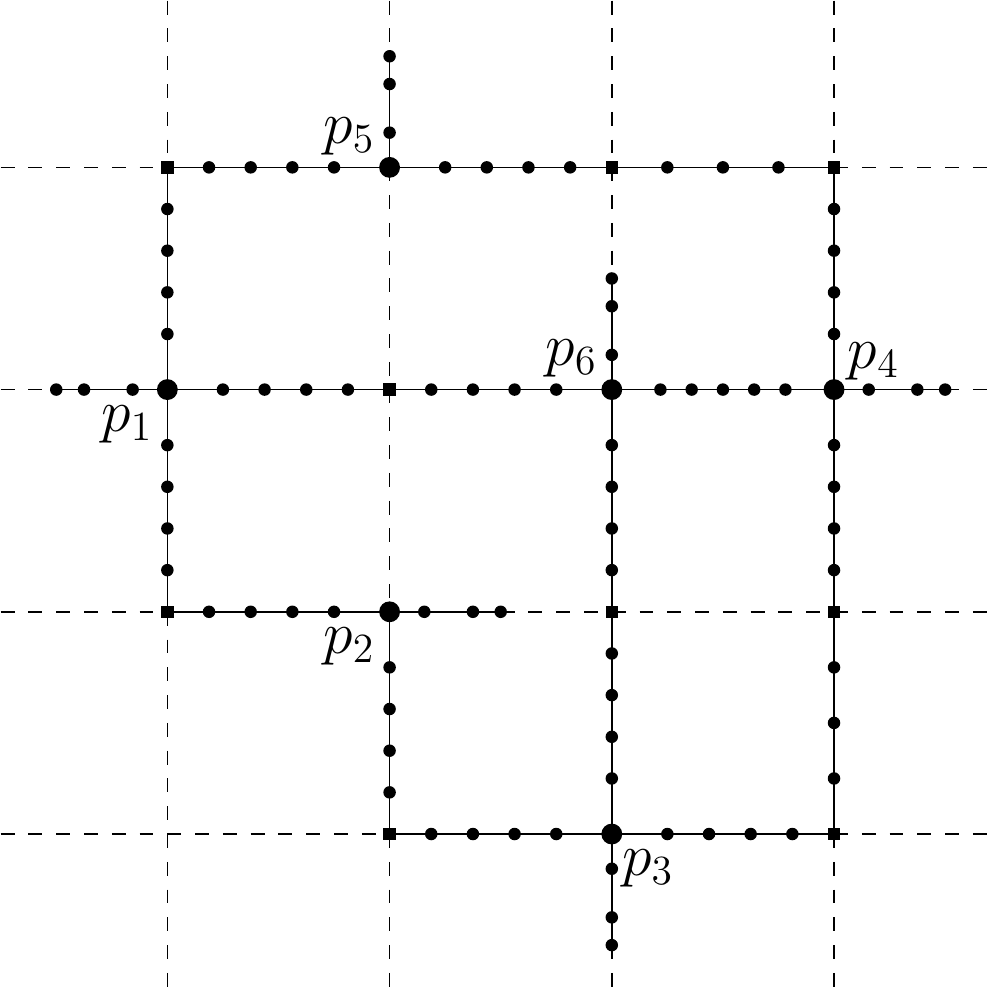}\\
  {(c)}
  \end{minipage}
    \caption{(a) A planar graph $G$ with maximum degree 3, 
  (b) its embedding on a $4 \times 4$ grid, and (c) construction of an UDG from the embedding.}\label{graph_grid}
  \end{figure}

\begin{lemma}\label{embedding-lemma}
 For a given \textsc{Vc-Pla} instance $G=(V,E)$  with at least one edge, an instance $G'=(V',E')$ of \textsc{Tds-Udg} can be constructed in polynomial-time.
\end{lemma}

\begin{proof}
Let $V = \{v_1, v_2, \ldots, v_n\}$ and $E = \{e_1, e_2, \ldots, e_m\}$. The construction of $G'$  from the graph $G$ is described in four steps. \\ 

\noindent
{\bf (a) Embedding: } We first embed $G$ into a planar grid of size $4 \times 4$ using Lemma \ref{cor:embedding}. On the embedding, each 
of the vertex $v_i \in V$ becomes grid point $p_i$ and each edge $e_j \in E$ become a finite sequence of connected line segment(s) of length four units along the grid lines. Assume that $\ell$ is the total number of line segments used in the embedding. We call the point $p_i$ corresponding to the vertex $v_i \in V$ ($i = 1, 2, \ldots, n$) in the embedding as \emph{vertex points} (see Fig. \ref{graph_grid}(a) and \ref{graph_grid}(b)). Let $N$ be the set of vertex points. Therefore, 
$N = \{p_i \mid v_i \in V\}$ and $|N| = |V| (= n)$.\\
 
\noindent
{\bf (b) Extra points: } In this step, we add some extra points on each of the $\ell$ line segments (obtained in embedding step) so that unit disks centered on these points and grid points (see embedding step) form an unit disk graph as follows: (a) for each edge $(p_i, p_j)$ with only one line segment i.e., length of the edge is 4 units, we add five points at distance 0.98, 1.49, 2, 2.51, 3.02 units from $p_i$ (see edge $(p_4, p_6)$ in  Fig. \ref{graph_grid}(c)), and (b) for each edge  $(p_i, p_j)$ with more than one segment i.e., length of the edge is greater than 4 units, (i) add a point on each of the grid point on the edge other than the vertex point and name it as grid point (see square points in Fig. \ref{graph_grid}(c)), and (ii) we add four points on each of the line segments connected with $p_i$ and $p_j$ at distances 1, 1.75, 2.5, 3.25 units from $p_i$ and $p_j$, and for other line segments we add three points at distance 1 units from each other excluding the \emph{grid points} (see the edge $(p_3, p_4)$ in Fig. \ref{graph_grid}(c)). Let $A$ be the set of all points added in this step. Therefore, $|A| = 4 \ell + m$, where $\ell$ is the total number of line segments in the embedding. \\

\noindent
{\bf (c) Support point: } Add a new line segment of length 1.4 units at each of the vertex point $p_i$ without coinciding with the line segments that had already been drawn in the embedding. Observe that addition of such line segment is possible without losing the planarity as the maximum degree of $G$ is 3. We add three points $x_i, y_i, z_i$ on each of these line segments at distances 0.3, 1.1, and 1.4 units from the corresponding vertex point $p_i$. Let $S$ be the set of all points added in this step. Therefore, $|S| = 3 n$. \\

\noindent
{\bf (d) Construction of UDG: } We construct a UDG $G'=(V', E')$, where $V' = N \cup A \cup S$ and $E' = \{(u', v'): u', v' \in V'$ and the Euclidean distance between $u'$ and $v'$ is at most 1 unit\} (see Fig. \ref{graph_grid}(c)).  From Lemma \ref{cor:embedding}, $\ell=O(n^2)$. Therefore both $|V'|$ and $|E'|$ are bounded by $O(n^2)$. Hence, $G'$ can be constructed in polynomial-time. 
 \end{proof}
 
\begin{theorem}\label{thm:main}
 \textsc{Tds-Udg} is NP-complete.
\end{theorem}

\begin{proof}
Let $T \subseteq V$ be an arbitrary subset of vertices and $k (>0)$ be an integer. Observe that, we can verify whether $T$ is a total dominating set 
such that $|T| \leq k$ or not in polynomial-time. Therefore, \textsc{Tds-Udg} $\in$ NP. 
  
To prove NP-hardness of \textsc{Tds-Udg}, we will use polynomial time reduction of \textsc{Vc-Pla} to it. We construct an instance $G'=(V',E')$ of 
\textsc{Tds-Udg} from an arbitrary instance $G=(V,E)$ of \textsc{Vc-Pla} in polynomial time using the steps mentioned in Lemma \ref{embedding-lemma}. 
Next, we prove the following claim to complete the proof of NP-hardness of \textsc{Tds-Udg}. 

{\bf Claim:} {\it $G$ has a vertex cover $C$ with $|C| \leq k$ if and only if $G'$ has a total dominating set $T$ with $|T| \leq k + 2\ell + 2n$}.

 \noindent {\bf Necessity:} Let $C\subseteq V$ be a vertex cover of $G$ such that $|C| \leq k$.  Let $N' = \{p_i \in N \mid v_i \in C\}$, i.e., $N'$ is the set of vertices (or vertex points) in $G'$ that correspond to the vertices in $C$. From each segment, we choose 2 vertices (extra points) from $A$ and corresponding to each vertex point, we choose 2 points (support points) from $S$, in the embedding. The set of chosen vertices, say $A' (\subseteq A)$, $S' (\subseteq S)$, 
together with $N'$ will form a TDS of desired cardinality in $G'$. We now discuss the process of obtaining the set $A'$. Initially $A' = \emptyset$. As $C$ is a vertex cover, every edge in $G$ has at least one of its end vertices in $C$. Let $(v_i, v_j)$ be an edge in $G$ and $v_i \in C$ (choose any of them arbitrarily if both $v_i$ and $v_j$ are in $C$). Note that the edge $(v_i, v_j)$ is represented as a sequence of line segments in the embedding. Start traversing the segments (of $(v_i, v_j)$) from $p_i$, where $p_i$ corresponds to $v_i$, and add two consecutive vertices by leaving two consecutive vertices in between starting from $p_i$ to $A'$  in the traversal (see $(p_4, p_5)$ in Fig. \ref{fig:proof} (b)). The red bold vertices are part of $A'$ while traversing from $p_4$).

\begin{figure}[!ht]
\centering
\begin{subfigure}[b]{.5\textwidth}  
  \centering
  \includegraphics[width=0.95\linewidth]{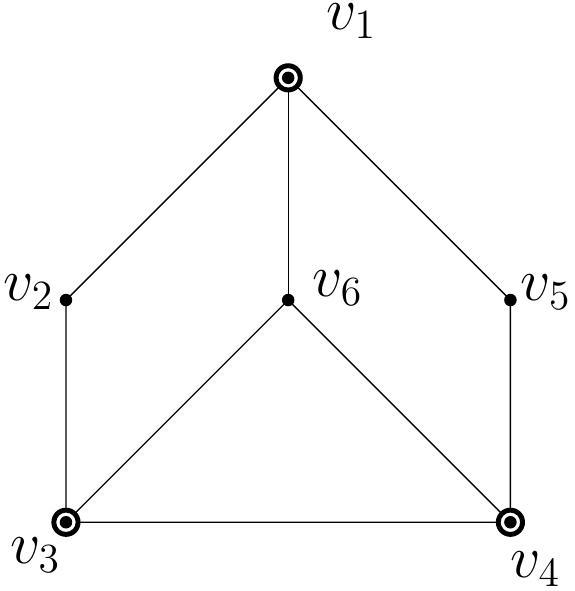}
  \caption{}  
\end{subfigure}%
\hfill
\begin{subfigure}[b]{.5\textwidth}  
  \centering
  \includegraphics[width=0.95\linewidth]{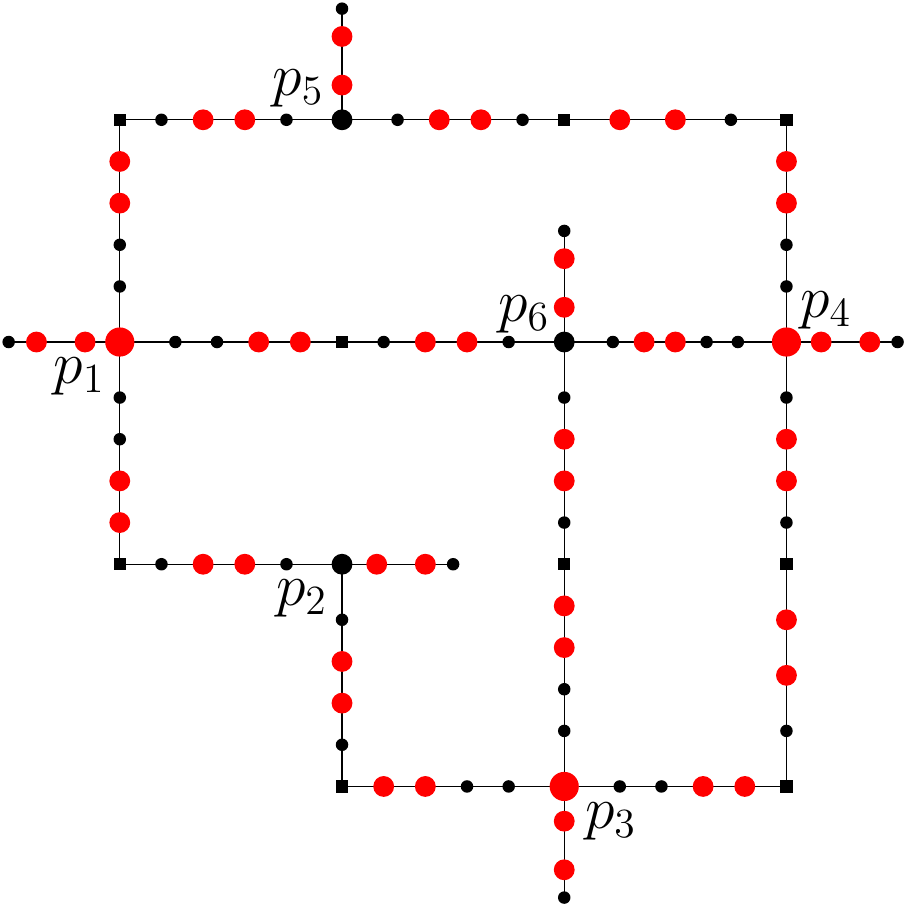}
  \caption{}  \label{fig:proofb}
\end{subfigure}
\caption{(a) A vertex cover $\{v_1,v_3,v_4\}$ in $G$, and (b) the construction of $A'$ in $G'$ (the tie between $v_3$ and $v_4$ is broken 
  by choosing $v_3$)}\label{fig:proof}
\end{figure}

Apply the above process to each edge in $G$. Observe that the cardinality of $A'$ is $2\ell$ as we have chosen 2 vertices from each segment in the embedding. Next, we choose $2n$ points from $S$ in $S'=\{x_i, y_i: p_i \in N\}$. Let $T = N' \cup A' \cup S'$.  Now, we argue that $T$ is a total dominating set in $G'$.

For each point $p_i \in N$, $p_i$ is dominated by $x_i$, $x_i$ is dominated by $y_i$, $y_i$ is dominated by $x_i$ and $z_i$ is dominated by $y_i$. So, the sets $N$ and $S$ satisfies total domination condition. Now it is remaining to prove that the set $A$ satisfies total domination condition. Observe the way we have chosen points from $A$ in $T$,  with a gap of two consecutive points, two consecutive points are chosen in $T$. For each point $p_i \in T$, $p_i$ dominates $N_G(p_i)\in A$ and the selected points of $A$ in $T$ can total dominate all the remaining points of $A$ (see the edge ($p_1,p_2$) in Fig. \ref{fig:proof} (b)).
 
Therefore, $T$ is a TDS in $G'$ and $|T| = |N'| + |A'| + |S'|\leq k + 2\ell + 2n$.
 
 \noindent {\bf Sufficiency:} Let $T \subseteq V'$ be a TDS of size at most $k + 2\ell + 2n$. 
 We prove that $G$ has a vertex cover of size at most $k$ with the help of the following claims.

 \begin{enumerate}[label=(\roman*)]
  \item Out of three support points associated with each $p_i \in N$, at least two points belongs to $T$, i.e., $|S \cap T|\geq 2n$.  
  \item Every segment in the embedding must contribute at least two points to $T$ and hence $|A \cap T| \geq 2\ell$, where $\ell$ is the total number of segments in the embedding.  
  \item If $p_i$ and $p_j$ correspond to end vertices of an edge $(v_i,v_j)$ in $G$, and if both $p_i,p_j$ are not in $T$, then there must be at least $2\ell' + 1$ vertices in  $T$ from the segment(s) representing the edge $(v_i,v_j)$, where $\ell'$ is the  number of segments representing the edge $(v_i,v_j)$ in the embedding.
 \end{enumerate}
 
 Claim (i) directly follows from the definition of total dominating set. Observe that we added points $x_i, y_i, z_i$ such that $p_i$
 is adjacent to $x_i$, $x_i$ is adjacent to $y_i$, and $y_i$ is adjacent to $z_i$ in $G'$, i.e., $\{(p_i,x_i),(x_i,y_i),(y_i,z_i)\}\subseteq E'$ for each $i$. 
 Hence, $y_i$ must be in $T$ as $y_i$ is the only vertex which can  dominate $z_i$ and either $x_i$ or $z_i$ must be in $T$ to dominate $y_i$.
Therefore, any total dominating set of $G'$ must contain two support points in $T$ out of three support points associated with 
$p_i, 1\leq i \leq n$, i.e., $|S \cap T|\geq 2n$.
 
Claim (ii) follows from the fact that only consecutive points are adjacent (in $G'$) on any segment in the embedding. Let $\eta$ be a segment in the embedding having vertices $q_i,q_{i+1},q_{i+2}$, and $q_{i+3}$. On contrary, assume that $\eta$ has only one of its vertices in $T$. Note that  only $q_i$ can not be in $T$.
 If $q_i$ present in $T$, then $q_{i+2}$ is not dominated by any point, which is a contradiction to the fact that $T$ is a TDS. If $q_{i+1}$ is the only point in $T$ then $q_{i+1}$ is not dominated by any point. If $q_{i+2}$ will be chosen as the only point from $\eta$ in $T$ then $q_{i+2}$ is not dominated by any other point  and finally if $q_{i+3}$ will be chosen then $q_{i+1}$ is not dominated by any other point. In all cases, we arrived at a contradiction.
 
Claim (iii) follows from the definition of total dominating set that any point chosen in the solution set dominates all its neighbours other than itself.  Here any point selected from a segment in $T$ has exactly two neighbours other than itself. So it can dominate at most 2 points. There are $\ell'$ segments between two node points $p_i$ and $p_j$ having $4\ell'+1$ number of points and both $p_i$ and $p_j$ are not in $T$. So, the minimum number of points required in $T$ to ensure total domination is $\ceil{\frac{4\ell'+1}{2}}=2\ell'+1$.

 Now, we will show that, by removing and/or replacing some vertices in $T$, a set of at most $k$ points from $N$ can be chosen such that the corresponding vertices form a vertex cover in $G$. The vertices in $S$ account for $2n$ vertices in $T$ (due to Claim (i)). Let $T = T \setminus S$ and $C = \{v_i \in V \mid p_i \in T \cap N\}$. If any edge $(v_i,v_j)$ in $G$ has none of its end vertices in $C$, then we do the following: consider the sequence of segments representing the edge $(v_i,v_j)$ in the embedding. Since, both $p_i$ and $p_j$  are not in $T$, there must exist a segment having three vertices in $T$ (due to Claim (iii)). Consider the segment having its three  vertices in $T$. Delete any one of the vertices on the segment and introduce $p_i$ (or $p_j$). Update $C$ and repeat the process till every edge has at least one of its end vertices in $C$. Due to Claim (ii), $C$ is a vertex cover in $G$ with $|C|\leq k$. Therefore, \textsc{Tds-Udg} is NP-hard. As  \textsc{Tds-Udg} $ \in$ NP   and \textsc{Tds-Udg} $\in$ NP-hard, \textsc{Tds-Udg} $\in$ NP-complete. 
\end{proof}

\section{Approximation Algorithm}\label{sec:apx}
In this section, we propose an algorithm to solve the TDS problem in UDGs, which produces an 8-factor approximation result. The worst case time complexity of our proposed algorithm is $O(n \log k)$, where $k$ is the size of the output of our algorithm. 

Let $P$ denote the set of $n$ points (center of the disks) given in the plane $\mathbb{R}^2$. We use $\Delta(S)$ to denote the unit disks centered at the points in a subset $S \subseteq P$. A dominating set in unit disk graphs (\textsc{Ds-Udg}) $D\subseteq P$ of the set of disks $\Delta(P)$ is said to be an independent \textsc{Ds-Udg} if for each pair $p$, $q\in D$, $p\not \in N_G[q]$.

The procedure of generating a \textsc{Tds-Udg} for a given points set $P$ in $\mathbb{R}^2$ is described in Algorithm \ref{algo:tds_apprx}.

\begin{algorithm}[!ht]
 \caption{Total dominating set in $P$}\label{algo:tds_apprx}
\begin{algorithmic}[1]
\Require A set of disks $\Delta(P)$
\Ensure A total dominating set $T$ of $\Delta(P)$
 \State $D \leftarrow \emptyset$, and $T \leftarrow \emptyset$
	  \While {($P\neq \emptyset$)}
	      \State choose an arbitrary point $p\in P$
	      \State $D \leftarrow D \cup p$; $P\leftarrow P\setminus N_G[p]$
	  \EndWhile
\For {every $p \in D$} \label{loop:sec_for_start}
      \If{$N_G(p)\cap T=\emptyset$}
      \State let $q \in N_G(p)$
	  \State $T = T \cup \{q\}$\label{line:add_v}
	  \EndIf
	  \EndFor
	  \State $T = T \cup D$ \label{line:add_w}
 \State \bf{return} $T$
 
\end{algorithmic}
\end{algorithm}

\begin{lemma} \label{lemma-3}
$T$ returned by Algorithm \ref{algo:tds_apprx} is a \textsc{Tds-Udg} for the set of unit disks $\Delta(P)$.
\end{lemma}

\begin{proof}
Before calculating a \textsc{Tds-Udg} $T$ the Algorithm \ref{algo:tds_apprx} calculate an independent \textsc{Ds-Udg} $D$ (see {\bf while} loop in line number 2 of the algorithm), which ensures domination for the set of unit disks $\Delta(P)$ and total domination for the points $P\setminus D$. Next, to obtain total domination in $D$, for each point $p\in D$ the algorithm ensures the existence of a point $q\in N_G(p)$ in $T$ (see {\bf for} loop in line number 5 of the algorithm). The selected points in $T$ along with $D$ ensures total domination for the set of unit disks $\Delta(P)$.
\end{proof}

\begin{lemma}\label{apx:apx}
$|T| \leq 8|OPT|$, where $OPT$ is a \textsc{Tds-Udg} for the unit disks $\Delta(P)$ of minimum size.
\end{lemma}

\begin{proof}
Consider an arbitrary point $p \in OPT$. 
As $OPT$ is a TDS of minimum size there must exist a point $q \in OPT$ such that $p\in N_G(q)$. The cardinality of $T$ follows from the fact that for every pair $p,q \in OPT$ such that $p\in N_G(q)$, there may exist at most 8 disks in an independent \textsc{Ds-Udg} $D$ of $\Delta(P)$ that can contain the points $p$ and/or $q$ (see Fig. \ref{fig:apx} for reference). For each disk at most 2 points can be chosen in the solution set $T$ to become a \textsc{Tds-Udg} (see line numbers \ref{loop:sec_for_start}-\ref{line:add_w} of Algorithm \ref{algo:tds_apprx}), which leads at most 16 points chosen by Algorithm \ref{algo:tds_apprx} against 2 points (namely, $p$ and $q$) chosen in optimal solution. Hence $|T|\leq 8|OPT|$. 
\end{proof}

\begin{figure}[!ht]
  \centering
  \includegraphics[width=0.65\linewidth]{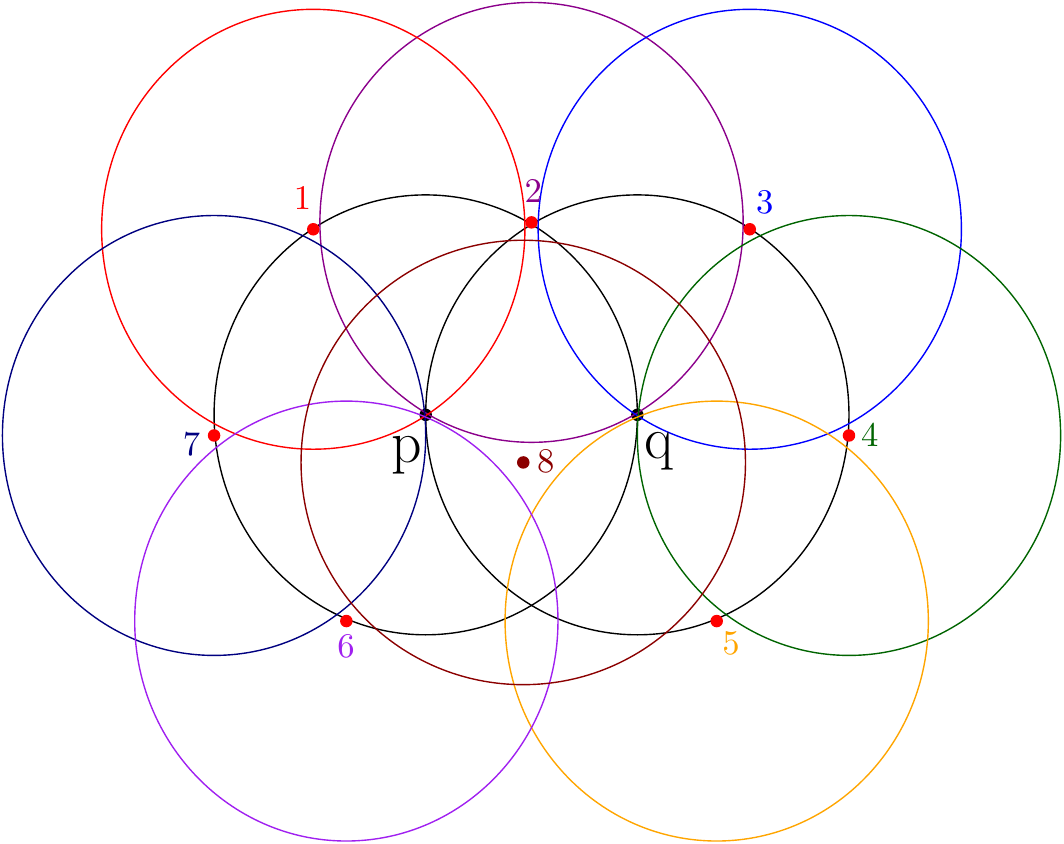}
\caption{Illustration of Lemma \ref{apx:apx}}\label{fig:apx}
\end{figure}
 
\begin{lemma}\label{apx:tc}
The worst case time require to generate a \textsc{Tds-Udg} for the set of disks $\Delta(P)$ by Algorithm \ref{algo:tds_apprx} is $O(n\log k)$, where $k$ is the size of the output.
\end{lemma}

\begin{proof}
We now describe the time complexity of the Algorithm \ref{algo:tds_apprx} for computing a \textsc{Tds-Udg} $T$ of $\Delta(P)$ as follows. 
Let us assume that $\cal R$ is an axis parallel rectangular region containing the points in $P$. We partition $\cal R$ into grid of size $1 \times 1$. 
A point $p_i=(x_i, y_i) \in P$ lies in the grid cell indexed by $[\lfloor{x_i}\rfloor, \lfloor{y_i}\rfloor]$ for $i = 1, 2, \ldots, n$. 
Each grid cell is attached with a list of points in $P$ that are belongs to that cell. We construct an independent dominating set $D$ for UDG 
corresponding to $\Delta(P)$. While considering a point $p_i \in P$, we inspect all members of $D$ which are attached to all 9 cells $[\alpha, \beta]$, where $\lfloor{x_i}\rfloor-1 \leq \alpha \leq \lfloor{x_i}\rfloor+1$ and $\lfloor{y_i} \rfloor-1 \leq \beta \leq \lfloor{y_i} \rfloor+1$. 
If there does not exists any unit disk $d$ in $D$ that contains the point $p_i$, we add $p_i$ in $D$. Observe that, at the end of considering all the points in $P$, $D$ will be an independent \textsc{Ds-Udg} for the set of disks in $\Delta(P)$. Initially, take $T=\emptyset$. Now, for each point $p \in D$, if there does not exist any point $q$ in $T$ such that $q \in N_G(p)$, then add $q$ in $T$, and existence of $q$ is guaranteed due to the input constraint of the problem.
Finally, update $T=T\cup D$. After ensuring the existence of a point $q \in N_G(p)$ for each point $p \in D$, observe that $T$ is a \textsc{Tds-Udg} for the set of disks in $\Delta(P)$. Note that, (i) a grid cell may contain at most 6 points in $T$, and (ii) the number of grid cells to be inspected while processing a point $p_i \in P$ is at most 9. We use a height balanced binary tree to store the indices of the grid cells containing a non-zero number of points in $T$. Thus, the time complexity for processing a point $p \in P$ is $O(\log k)$, where $k=|T|$ and $|P|=n$. Thus, the result.
\end{proof}

\begin{theorem}
Algorithm \ref{algo:tds_apprx} is an 8-factor approximation algorithm for the \textsc{Tds-Udg} problem. The running time of the algorithm is $O(n\log k)$, where $k$ is the size of the output. 
\end{theorem}

\begin{proof}
The proof of the theorem follows from Lemma \ref{apx:apx} and Lemma \ref{apx:tc}. 
\end{proof}

 \section{Approximation Scheme}\label{sec:ptas}
 In this section, we propose a polynomial time approximation scheme (PTAS) for the TDS problem in unit disk graphs. We use shifting strategy  \cite{hochbaum} technique to propose a PTAS. Let $P$ be a point set (centers of the disks) given in a rectangular region ${\cal R}$ along with a fixed integer $k \geq 1$.

 We use a two-level nested shifting strategy to propose a PTAS for the said problem.
 The first level of shifting strategy applied in the horizontal direction on ${\cal R}$. 
 There are $k$ iterations in the first level and the $i$-th iteration $(1\leq i\leq k)$ partition the region ${\cal R}$ into many horizontal strips, where the first strip is of width $2i$, and remaining strips other than the last strip are of width $2k$. 
 The width of last strip may be less than $2k$.  
 Without loss of generality, assume that each point lying on the left boundary of a strip belong to its left adjacent strip.
 Now consider all the non-empty horizontal strip $H$, and apply second level of shifting strategy on the vertical direction.
  In the second level of shifting strategy, the $j$-th iteration $(1\leq j\leq k)$ partition each non-empty horizontal strip $H$ into square/rectangular cells of size $2j\times \ell$ for the first cell and $2k\times \ell$ for all other cells, where $\ell$ defines the width of the strip $H$ ($\ell=2i$ for the first strip and $\ell=2k$ for all other strips except last strip).

We consider each non-empty $2k \times 2k$ squares (conceptually extending the smaller cells into  $2k \times 2k$ square) and find the optimal solution of each squares.  
The union of the optimal solution of each $2k \times 2k$ squares give a feasible solution of each strip $H$. 
Finally, we take the union of solutions of  each non-empty horizontal strips to get a feasible solution of the problem in a single iteration. 
In the same process, we get the feasible solutions of all the iterations in the first level. 
We report the solution $T$, having minimum cardinality among all the solutions generated in each iterations as the solution of  the \textsc{Tds-Udg} problem.

  Now, we discuss the procedure of getting optimal solution from each $2k \times 2k$ square. 
  We first partition the cell of size $2k \times 2k$ into $(\ceil{2\sqrt{2}k})^2$ sub-cells. 
  The size of each sub-cell is $\frac{1}{\sqrt{2}}\times \frac{1}{\sqrt{2}}$.
   Observe that, choosing any two points inside a sub-cell of size $\frac{1}{\sqrt{2}}\times \frac{1}{\sqrt{2}}$ ensures total domination for all unit disk centered in that sub-cell. Hence, the maximum number of points required to ensure total domination in a square of size $2k\times 2k$ is $2(\ceil{2\sqrt{2}k})^2$.  
   Therefore, we have to check all possible combinations of points upto $2(\ceil{2\sqrt{2}k})^2$ to get an optimal solution in a cell $\chi$ of size $2k\times 2k$. 
   Note that, along with the points inside a cell $\chi$, the points within 1 unit apart from $\chi$ is also plays a crucial role to get an optimum solution of $\chi$. 
   Let $n_\chi$ be the number of points in $P$ whose corresponding disks has a portion in the cell $\chi$ ($n_\chi$ includes the points inside $\chi$ along with the points within 1 unit apart from $\chi$). 
     Then, we have to choose at most $O(n_\chi^{2(\ceil{2\sqrt{2}k})^2})$ combinations of points for getting the optimum solution for the \textsc{Tds-Udg} problem in a cell $\chi$ of size $2k\times 2k$. 
     Since the points in $P$ centered in a cell is disjoint from that of the other cells, and a point in $P$ can participate in computing the optimum solution of at most 9 cells, we have the following result.
\begin{lemma}\label{ptas:tc}
The total time required for the $(i,j)$-th iteration of the algorithm is $O(n^{2(\ceil{2\sqrt{2}k})^2})$.
\end{lemma}
\begin{proof}
The feasible solution of the $(i,j)$-th iteration is the union of the optimum solutions of all the cells constructed in that iteration. 
Finally, the algorithm returns the minimum among the $k^2$ feasible solutions corresponding to $k^2$ iterations.
\end{proof}
\begin{theorem}
Given a set $P$ of $n$ points (center of the unit disks) in ${\cal R}$ and an integer $k \geq 1$, a total dominating set of size at most $(1+\frac{1}{k})^2 \times |OPT|$ can be computed in $O(k^2n^{{2(\ceil{2\sqrt{2}k})^2})}$ time, where $OPT$ is the optimum solution.
\end{theorem}
 \begin{proof}
 Using the shifting strategy analysis given by Hochbaum and maass \cite{hochbaum}, we analyze the approximation factor of our algorithm. 
 Let $OPT$ be an optimum solution for the \textsc{Tds-Udg} problem for the point set $P$, and $OPT'\subseteq OPT$ be such points chosen in $OPT$, which total dominates the points outside the boundary  of all the cells in an $(i,j)$-th iteration. 
 Let $T$ be a solution obtained by our algorithm in an iteration.
  Then, $|T|\leq |OPT|+|OPT'|$. For all the iterations of $(i,j)$ ($1\leq i,j \leq k$), we have 
     $\sum\limits_{i=1}^{k}\sum\limits_{j=1}^{k} |T|\leq k^2|OPT|+\sum\limits_{i=1}^{k}\sum\limits_{j=1}^{k}|OPT'|$.\\
Since any point from a cell $\chi$ chosen in $OPT$ can dominate points from no more than one horizontal strip (or vertical strip), and at most $k$ times each horizontal (or vertical) boundary appears throughout the algorithm, we have\\
 $\sum\limits_{i=1}^{k}\sum\limits_{j=1}^{k}|OPT'|\leq k|OPT|+k|OPT|$.\\
 Thus,\\
 $\sum\limits_{i=1}^{k}\sum\limits_{j=1}^{k}|T|\leq k^2|OPT|+2k|OPT|=(k^2+2k)|OPT|$.\\ Thus, $min\sum\limits_{i=1}^{k}\sum\limits_{j=1}^{k}|T|\leq (1+\frac{1}{k})^2 \times |OPT|$.
\\ the time complexity result follows from Lemma \ref{ptas:tc}. 
 \end{proof}

\section{Conclusion}\label{conclusion}
In this article, we have considered the minimum total dominating set (TDS) problem in unit disk graphs. We showed that the TDS problem is NP-hard. 
We proposed an almost linear time 8-factor approximation algorithm and a PTAS for the same problem.

\bibliographystyle{abbrv}

\begin{thebibliography}{10}

\bibitem{amos2012total}
D.~Amos and E.~DeLaVina.
\newblock On Total Domination in Graphs.
\newblock {\em Senior Project, University of Houston Downtown}, 2012.

\bibitem{brigham2000connected}
R.~C. Brigham, J.~R. Carrington, and R.~P. Vitray.
\newblock Connected Graphs with Maximum Total Domination Number.
\newblock {\em Journal of Combinatorial Mathematics and Combinatorial
  Computing}, 34:81--96, 2000.


\bibitem{chellali2006note}
M.~Chellali and T.~W. Haynes.
\newblock A Note on the Total Domination Number of a Tree.
\newblock {\em Journal of Combinatorial Mathematics and Combinatorial
  Computing}, 58:189, 2006.

\bibitem{chvatal1992small}
V.~Chv{\'a}tal and C.~McDiarmid.
\newblock Small Transversals in Hypergraphs.
\newblock {\em Combinatorica}, 12(1):19--26, 1992.

\bibitem{cockayne1980}
E.~J. Cockayne, R.~Dawes, and S.~T. Hedetniemi.
\newblock Total Domination in Graphs.
\newblock {\em Networks}, 10(3):211--219, 1980.

\bibitem{cockayne2003vertices}
E.~J. Cockayne, M.~A. Henning, and C.~M. Mynhardt.
\newblock Vertices Contained in all or in no Minimum Total Dominating Set of a
  Tree.
\newblock {\em Discrete mathematics}, 260(1-3):37--44, 2003.

\bibitem{delavicna2007some}
E.~DeLaVi{\c{n}}a, Q.~Liu, R.~Pepper, B.~Waller, and D.~B. West.
\newblock Some Conjectures of Graffiti. pc on Total Domination.
\newblock 2007.

\bibitem{dorfling2006domination}
M.~Dorfling, W.~Goddard, and M.~A. Henning.
\newblock Domination in Planar Graphs with Small Diameter ii.
\newblock 2006.

\bibitem{garey}
M.~R. Garey and D.~S. Johnson.
\newblock {\em Computers and Intractability: a Guide to the Theory of
  {NP}-completeness}.
\newblock Freeman, 1979.

\bibitem{haynes2002trees}
T.~Haynes and M.~Henning.
\newblock Trees with Unique Minimum Total Dominating Sets.
\newblock {\em Discussiones Mathematicae Graph Theory}, 22(2):233--246, 2002.

\bibitem{haynes1998fundamentals}
T.~W. Haynes, S.~Hedetniemi, and P.~Slater.
\newblock {\em Fundamentals of Domination in Graphs}.
\newblock CRC press, 1998.

\bibitem{haynes1998domination}
T.~W. Haynes, S.~T. Hedetniemi, and P.~J. Slater.
\newblock Domination in Graphs (Advanced Topics) Marcel Dekker Publications.
\newblock {\em New York}, 1998.

\bibitem{henning2009survey}
M.~A. Henning.
\newblock A Survey of Selected Recent Results on Total Domination in Graphs.
\newblock {\em Discrete Mathematics}, 309(1):32--63, 2009.

\bibitem{henning2013total}
M.~A. Henning and A.~Yeo.
\newblock {\em Total Domination in Graphs}.
\newblock Springer, 2013.

\bibitem{hochbaum}
D.~S. Hochbaum and W.~Maass.
\newblock Approximation Schemes for Covering and Packing Problems in Image
  Processing and {VLSI}.
\newblock {\em Journal of the ACM (JACM)}, 32(1):130--136, 1985.

\bibitem{hopcroft}
J.~Hopcroft and R.~Tarjan.
\newblock Efficient Planarity Testing.
\newblock {\em Journal of the ACM (JACM)}, 21(4):549--568, 1974.

\bibitem{itai}
A.~Itai, C.~H. Papadimitriou, and J.~L. Szwarcfiter.
\newblock Hamilton Paths in Grid Graphs.
\newblock {\em SIAM Journal on Computing}, 11(4):676--686, 1982.

\bibitem{laskar1984algorithmic}
R.~Laskar, J.~Pfaff, S.~Hedetniemi, and S.~Hedetniemi.
\newblock On the Algorithmic Complexity of Total Domination.
\newblock {\em SIAM Journal on Algebraic Discrete Methods}, 5(3):420--425, 1984.

\bibitem{mishra2020DAM}
Pawan K. Mishra, Sangram K. Jena, Gautam K. Das  and S. V. Rao.
\newblock Capacitated discrete Unit Disk Cover.
\newblock {\em Discrete Applied Mathematics}, 285: 242-251, 2020. 

\bibitem{pfaff1983np}
J.~Pfaff, R.~Laskar, and S.~Hedetniemi.
\newblock Np-completeness of Total and Connected Domination and Irredundance
  for bipartite graphs.
\newblock In {\em Tech. Rept. 428}. Clemson University Clemson, SC, 1983.

\bibitem{sun1995upper}
L.~Sun.
\newblock An Upper Bound for the Total Domination Number.
\newblock {\em J. Beijing Inst. Tech}, 4(2):111--114, 1995.

\bibitem{thomasse2007total}
S.~Thomass{\'e} and A.~Yeo.
\newblock Total Domination of Graphs and Small Transversals of Hypergraphs.
\newblock {\em Combinatorica}, 27(4):473--487, 2007.

\bibitem{tuza1990covering}
Z.~Tuza.
\newblock Covering all Cliques of a Graph.
\newblock {\em Discrete Mathematics}, 86(1-3):117--126, 1990.


\end{thebibliography}

\end{document}